\documentclass[conference]{IEEEtran}

\ifCLASSINFOpdf

\else

\fi

\usepackage{amssymb,amsmath,dsfont,mathrsfs,theoremref}
\usepackage{graphicx,enumerate,color,booktabs,float}
\usepackage{amsthm}
\usepackage{hyperref}
\usepackage{flushend}
\usepackage{cleveref}
\DeclareMathOperator*{\argmax}{arg\,max}
\DeclareMathOperator*{\argmin}{arg\,min}
\newcommand{\pr}{\mathbb P}
\newcommand{\ex}{\mathbb E}
\newcommand{\ind}{\operatorname{\mathds{1}}}
\newcommand{\mc}{\mathcal}
\allowdisplaybreaks
\theoremstyle{plain}
\newtheorem{theorem}{Theorem}
\newtheorem{corollary}{Corollary}[theorem]
\newtheorem{prop}{Proposition}

\theoremstyle{definition}
\newtheorem{defn}{Definition}

\theoremstyle{remark}	
\newtheorem{remark}{Remark}

\begin{document}

	\title{Optimal De-Anonymization in Random Graphs with Community Structure}

	\author{\IEEEauthorblockN{Efe Onaran, Siddharth Garg, Elza Erkip}		
		\IEEEauthorblockA{NYU Tandon School of Engineering}
		\IEEEauthorblockA{eonaran@nyu.edu, siddharth.garg@nyu.edu, elza@nyu.edu}

	}

	\maketitle
	\begin{abstract}
	Anonymized social network graphs published for academic or 
	advertisement purposes are subject to de-anonymization attacks by leveraging 
	side information in the form of a second, public social network graph correlated 
	with the anonymized graph. 
	This is because the two are	
	from the same underlying graph of true social relationships. In this paper, 
	the maximum a posteriori (MAP) estimates of user identities for 
	the anonymized graph are characterized and sufficient conditions for successful de-anonymization 
	for underlying graphs with community structure are provided. The results generalize 
	prior work that assumed underlying graphs of Erd\H{o}s-R\'enyi type, and prove the optimality of the attack strategy adopted in the literature.  
	\end{abstract}
	
	\IEEEpeerreviewmaketitle

	\section{Introduction}
	\label{intro}	
Privacy of users in social networking sites is an important
concern, especially with the rising popularity of social
networks. Social network data is often intentionally (and
sometimes accidentally) revealed for academic and advertisement
purposes. Before such revelations, names and other
identifying information about users are usually omitted (or
randomized) in an effort to preserve the anonymity of the
users. However, removing user names  from published
data may not in itself be sufficient to preserve privacy. This is
because an attacker might be able to use publicly available side
information correlated with the anonymized data to recover
user identities. 
For instance, anonymized
Netflix data of user viewing preferences was de-anonymized
using a publicly available IMDB database~\cite{narayanan2008robust}.

In this work, we focus on the problem of privacy in social
networks when the data published by the networking site is an
anonymized graph of user connectivity. That is, each vertex in
the graph is a user, and users that are connected (for instance,
friends) share an undirected edge. The graph is anonymized
by removing or randomizing the labels corresponding to each
vertex. However, as in the Netflix/IMDB example above, an
attacker might be able to recover some or all of the user
identities using data from a second social network whose
graph of user connectivity is publicly available. For example, Narayanan and Shmatikov 
demonstrated an attack in which 
anonymized Twitter users were
identified using Flickr data~\cite{narayanan2009anonymizing}. 

The success of social network de-anonymization is premised on the 
observation that data
from two social networks can be reasonably expected to
be correlated assuming that the two graphs, $g_1$ and $g_2$, are
independently sampled (with probability $s$) from the same
underlying graph, $g$, which represents the \emph{true} relationships
between all users. 
In recent work~\cite{grossglauser11}, Pedarsani and Grossglauser
have theoretically analyzed the social network de-anonymization 
by assuming that the underlying
graph, $g$, is an Erd\H{o}s-R\'enyi graph with edge probability $p$.
In particular, \cite{grossglauser11} establishes \emph{sufficient} conditions
on the edge probability, $p$, and sampling probability, $s$, that
guarantee perfect de-anonymization, that is, each user in the
anonymized network is correctly identified.

There are several important directions in which the prior work can be extended. First, in obtaining their
results, Pedarsani and Grossglauser~\cite{grossglauser11} and succeeding work~\cite{kazemi2015can,lyzinski2014seeded} assume that the attacker
matches nodes in the anonymized graph to those in the public
graph so as to minimize a certain cost function, which is the number of mismatched edges. However, the
authors do not provide a rationale for this choice of cost
function, or prove the optimality of the attacker's strategy. Second, it has been empirically shown that graphs of social interaction are
not, in fact, Erd\H{o}s-R\'enyi 
and instead have community structure~\cite{girvan2002community}.
A natural question, therefore, is what are the attacker's
capabilities when the statistics of the underlying graph reflect
community structure? The goal of our work is to address these
questions.

Specifically, in this paper, we make the following contributions.
First, we derive the maximum a posteriori probability
(MAP) matching of nodes in the anonymized graph
to the public graph assuming an underlying graph $g$ with
community structure. For the special case of a graph with
a single community, which is simply an Erd\H{o}s-R\'enyi graph,
we establish the equivalence between the MAP matching and
the matching that minimizes Pedarsani and Grossglauser's cost
function, establishing the optimality of the
attacker's strategy in that and several succeeding works.
Second, we determine conditions on the parameters of the
underlying graph and the sampling probability that guarantee
perfect de-anonymization for the general case of graphs with
community structure.

We note that when $g$ has community structure, 
the de-anonymization problem has some relation 
to the well known community detection problem~\cite{commdetect}. 
Hence the problem can be 
viewed as community detection with side information, although we 
note that recovering the community labels 
in $g_2$ is only a necessary, but not 
sufficient, condition for 
de-anonymization. In the rest of this paper, we 
will focus on an attack setting in which the attacker's goal is to determine user identities only, assuming that the community labels in $g_1$ and $g_2$ are known.

	\section{System Model}		
	Let $g = (V,E_g)$ 
	be an undirected graph of true relationships 
	between $n$ users.  The number of users is
	$|V|=n$, and the edge set is
	$E_g \subseteq V \times V$.    
	 We assume nodes of $ g $ are partitioned into disjoint subsets of $ k \geq 1 $ communities, $C_1, C_2, \ldots, C_k$. The number of nodes in community 
	 $C_i$ is $|C_i| = n_i$. An example of $g$ for the 
	 two community case is shown in Figure~\ref{fig:model}.
	 
	 \begin{figure}[!ht]
  \centering
    \includegraphics[width=0.45\textwidth]{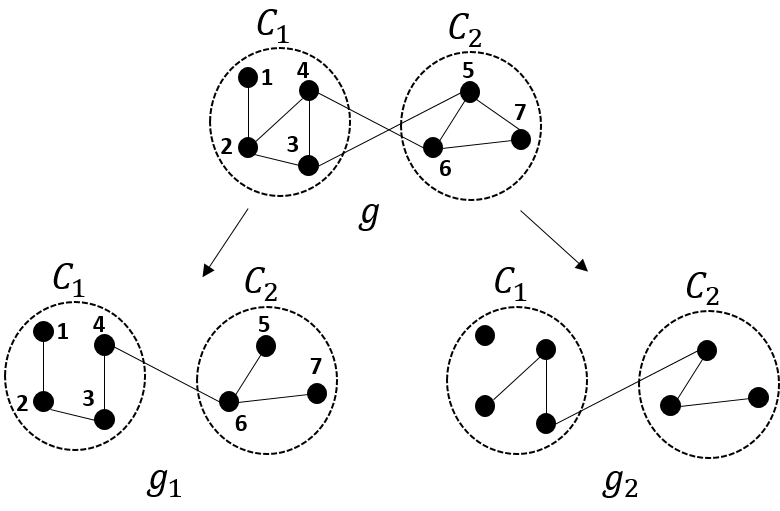}
    \caption{System model with the underlying 
    graph $g$ with $2$ communities, the public graph $g_1$ and the anonymized graph $g_2$.}
    \label{fig:model}
    \end{figure}
    
	 Edges in $E_g$ are drawn independently at random as follows: nodes $u \in C_i$ and $v \in C_{j}$ are connected with probability $p_{ij}$. 
In literature, this model is referred to as the 
\emph{stochastic block model}~\cite{girvan2002community}, and is simply the Erd\H{o}s-R\'enyi graph model for $k=1$. 
	%

From underlying graph of true relationships between 
users, $g$,  
we obtain two independent samples $ g_1 = (V, E_{g_1}) $ and $ g_2 = (V, E_{g_2})$, as shown with an example in Figure~\ref{fig:model}, 
that 
represent the connectivity graphs of two different 
social network platforms. As in \cite{grossglauser11}, we 
assume that only edges are sampled and not nodes, i.e., both social networks capture the same population of 
individuals. In particular:
	\begin{equation*}
	\pr\{(u,v)\in E_{g_i}\}=\begin{cases}
	s &\text{if}\; (u,v)\in E_{g}\\
	0 &\text{otherwise}
	\end{cases}
	\end{equation*} 
	for $ i=1,2 $  where we further assume that each edge is sampled independently. 
	
	In the context of the 
	de-anonymization problem, we assume that all 
	the labels of the nodes of $g_1$ are given and the attacker's goal is to recover the labels of $g_2$ given $g_1$. {However, as noted in Section~\ref{intro}, 
	we will assume that the community assignments 
	are known in both $g_1$ and $g_2$.}  
	This problem
	is equivalent to finding a 
	matching between the node sets of two graphs. A 
	matching of $ g_1 $'s nodes to $g_2$'s nodes, which is also a permutation over $ [n] $, 
	will be denoted as $ \pi : V \rightarrow V $, where $\Pi$ is the set of all possible permutations that is compatible with the community assignments. 
	
	With some abuse of notation, for an 
	edge $e \in E_{g_1}$, we will allow $\pi(e)$ to 
	denote the image of $e$ in $E_{g_2}$. That is, for 
	$(u,v) \in  E_{g_1}$, $\pi((u,v)) = (\pi(u),\pi(v))$.
	Given community labels, $E^{ij}_{g_k}$ refers to the 
	subset of edges in $g_k$ ($k \in \{1,2\}$) 
	that lie between nodes in community $C_{i}$ and $C_{j}$ and similar definition follows for $E^{ij}_{g}$.   
	 
	    
	\section{Maximum a Posteriori De-Anonymization} \label{MAP}
	In this section, we derive the 
	MAP estimate of $\pi$ (the labeling of nodes in $g_2$), assuming an a priori uniform distribution over 
	$ \pi $, 
	given realizations of $g_1$ and $g_2$. In particular, we will show that the the MAP estimate corresponds to the 
	matching that minimizes the cost function, $\Delta_{\pi}$
	that we define below.
\begin{defn}
\begin{equation}
	\Delta_{\pi}=\sum_{i\leq j}^{k} \omega_{ij}\hspace{-0.4em} \left( \sum_{e\in E_{g_1}^{ij}} \ind\{\pi(e)\notin E_{g_2}^{ij}\} +\hspace{-0.8em}\sum_{e\in E_{g_2}^{ij}} \ind\{\pi^{-1}(e)\notin E_{g_1}^{ij}\}\hspace{-0.3em}\right) \label{0431}
\end{equation}
where
\begin{equation*}
	\omega_{ij}= \log \left( \frac{1-p_{ij} s(2-s)}{p_{ij}(1-s)^2}\right)\;, 
	\end{equation*}
	and $ \sum_{i\leq j}^{k} $ is the short-hand notation for $ \sum_{1\leq i\leq j\leq k} $.
\end{defn}

By virtue of being the MAP estimate, 
	minimizing $\Delta_{\pi}$ corresponds to the attacker's optimal 
	strategy that minimizes error probability~\cite{salehi2007digital}, i.e., 
	the probability of obtaining 
	an incorrect matching. 
	On the other hand, if we set all the weights $\omega_{ij}$ in Equation~\ref{0431} 
to $1$, $\Delta_{\pi}$ above would correspond to the 
``edge mismatch" 
cost function that \cite{grossglauser11} and subsequent works~\cite{lyzinski2014seeded,kazemi2015can}
minimize. This strategy is
sub-optimal, in general, for community 
structured graphs.
 


Now, we define the MAP estimate of $\pi$ given 
$g_1$ and $g_2$ as:
\begin{align*}
	\text{MAP}(g_1,g_2)&=  \argmax_{\pi\in \Pi} p(\pi|g_1,g_2) \; ,
\end{align*}
which can be further written as:
\begin{align*}
	\text{MAP}(g_1,g_2)=\argmax_{\pi\in \Pi} \sum_{g\in \mc G_{\pi}} p(g,\pi|g_1,g_2) \; , 
\end{align*}
where $\mc G_{\pi}$ is the set of all 
underlying graphs that are consistent 
with $g_1$ and $g_2$ given $\pi$.  


\begin{theorem}\label{mapthm}
Assuming $p_{ij} < 1/2$ for all $i,j \in [1,k]$
	\begin{align*}
		\text{MAP}(g_1,g_2)=\argmin_{\pi\in\Pi} \Delta_{\pi}
	\end{align*}
\end{theorem} 
\begin{proof}		
		 Let us start with elaborating $ \text{MAP}(g_1,g_2) $,
	\begin{align}
		&\text{MAP}(g_1,g_2) \nonumber\\
		&=\argmax_{\pi\in \Pi} \sum_{g\in \mc G_{\pi}}  p(g_1|g)\cdot p(g_2|g,\pi)\cdot p(g) \label{1538}\\
		\begin{split}	
			&=\textstyle\argmax_{\pi\in \Pi} \sum_{g\in \mc G_{\pi}}  \prod_{i\leq j}^{k} (1-s)^{|E_{g}^{ij}|-|E_{g_1}^{ij}|}s^{|E_{g_1}^{ij}|} \\
			&\textstyle\hspace{9em}\cdot \prod_{i\leq j}^{k} (1-s)^{|E_{g}^{ij}|-|E_{g_2}^{ij}|}s^{|E_{g_2}^{ij}|}\\
			&\textstyle\hspace{9em}\cdot \prod_{i\leq j}^{k} p_{ij}^{|E_{g}^{ij}|} (1-p_{ij})^{N_{}^{ij}-|E_{g}^{ij}|}	
		\label{1643} \end{split} \\
		\begin{split}		
			&=\argmax_{\pi\in \Pi}\left( \prod_{i\leq j}^{k}\left(\frac{s}{1-s}\right)^{|E_{g_1}^{ij}|+|E_{g_2}^{ij}|}(1-p_{ij})^{N_{}^{ij}} \right) \\ 
			&\hspace{4.4em}\cdot\left(\sum_{g\in \mc G_{\pi}} \prod_{i\leq j}^{k} \left( \frac{p_{ij}(1-s)^2}{1-p_{ij}}\right)^{|E_{g}^{ij}|}\right)	 
		\end{split}\nonumber\\
		&=\argmax_{\pi\in \Pi}  \sum_{g\in \mc G_{\pi}} \prod_{i\leq j}^{k} \left( \frac{p_{ij}(1-s)^2}{1-p_{ij}}\right)^{|E_{g}^{ij}|}\label{0341} 
	\end{align}
	In \eqref{1538} we used the facts
	 that distribution of $ g_1 $ given $ g $ 
	is independent of the labeling of $ g_2 $, and $ \pi $ has uniform prior distribution.
	 In \eqref{1643}, $ N^{ij} $
	 is the number of node pairs from communities $ C_i $ and $ C_j $, and  while writing \eqref{0341} we used 
	 the observation that 
		$\sum_{i\leq j}^{k} \left({|E_{g_1}^{ij}|+|E_{g_2}^{ij}|}\right)={|E_{g_1}|+|E_{g_2}|}$
	is constant given $g_1$ and $g_2$.
	
	Now let $ g^*_{\pi} $ denote the graph having the smallest number of edges in $ \mc G_\pi $, i.e, $g^*_{\pi}  = (V, E_{g_1} \cup \pi(E_{g_2}))$. 
	Note that $ \mc G_\pi $
	 consists of all graphs whose edge sets are supersets 
	 of  $ g^*_{\pi} $. Summing over all graphs in $ \mc G_\pi $, we get:
	\begin{align}
		\begin{split}
			&\text{MAP}(g_1,g_2)=\argmax_{\pi\in \Pi}\left\{ \prod_{i\leq j}^{k} \left( \frac{p_{ij}(1-s)^2}{1-p_{ij}}\right)^{|E_{ g^*_{\pi}}^{ij}|}  \right.\\
			&\left. \cdot \prod_{i\leq j}^{k}\; \sum_{a_{ij}=0}^{N^{}_{ij}-|E_{ g^*_{\pi}}^{ij}|} {N^{}_{ij}-|E_{ g^*_{\pi}}^{ij}| \choose a_{ij}} \left( \frac{p_{ij}(1-s)^2}{1-p_{ij}}\right)^{a_{ij}}   \right\} \label{0338}
		\end{split} 
	\end{align}
	%
	Noting that sum in \eqref{0338} is a binomial sum, we can further write \eqref{0338} as: 
	\begin{align}
		&\textstyle\argmax_{\pi\in \Pi} \prod_{i\leq j}^{k} \left( \frac{p_{ij}(1-s)^2}{1-p_{ij} s(2-s)}\right)^{|E_{ g^*_{\pi}}^{ij}|} \nonumber\\
		&=\textstyle	\argmin_{\pi\in \Pi}
		 \sum_{i\leq j}^{k} |E_{ g^*_{\pi}}^{ij}| \log \left( \frac{1-p_{ij} s(2-s)}{p_{ij}(1-s)^2}\right)  \label{0358}
	\end{align}
Now  we observe that 			
	\begin{align*}
		&|E_{ g^*_{\pi}}^{ij}| =\frac{1}{2}\left(|E_{ g_1}^{ij}|+|E_{ g_2}^{ij}|+\right.  \\ 
		&\left.\sum_{e\in E_{g_1}^{ij}} \ind\{\pi(e)\notin E_{g_2}^{ij}\} +\sum_{e\in E_{g_2}^{ij}} \ind\{\pi^{-1}(e)\notin E_{g_1}^{ij}\}\right)\;,	    	
	\end{align*}	
	which allows us to write \eqref{0358} as	
	\begin{align*}
		\begin{split}
			&\argmin_{\pi\in \Pi}\textstyle \sum_{i\leq j}^{k} \omega_{ij}\hspace{-0.3em} \left( 
			\textstyle\sum_{e\in E_{g_1}^{ij}} \ind\{\pi(e)\notin E_{g_2}^{ij}\}\hspace{-0.25em}\right.\\
			&\left. \hspace{7.5em}+\textstyle\sum_{e\in E_{g_2}^{ij}} \ind\{\pi^{-1}(e)\notin E_{g_1}^{ij}\}\hspace{-0.3em}\right)
		\end{split}\\
		&=\argmin_{\pi\in \Pi} \Delta_{\pi}
	\end{align*}
	since $|E_{ g_1}^{ij}|+|E_{ g_2}^{ij}|$ does not depend on $ \pi $. 
	 \end{proof}
\begin{corollary} \label{singlecomm}
If $g$ is an Erd\H{o}s-R\'enyi graph with 
edge probability $p < 1/2$
\begin{equation}
\begin{split}
&\text{MAP}(g_1,g_2)=\\
&\hspace{-1em}\argmin_{\pi\in\Pi} \left( \sum_{e\in E_{g_1}} \ind\{\pi(e)\notin E_{g_2}\} +\hspace{-0.6em}\sum_{e\in E_{g_2}} \ind\{\pi^{-1}(e)\notin E_{g_1}\}\hspace{-0.3em}\right)\label{eq:single}
\end{split}
\end{equation}
\end{corollary} 
\begin{proof}
Set $k=1$ in Theorem~\ref{mapthm} and note that 
for $p_{11} < 1/2$, $\omega_{11} >0$.
\end{proof}

\begin{remark}
The cost function in \eqref{eq:single} 
is the same as the ``edge mismatch" cost function 
that \cite{grossglauser11} and subsequent works \cite{lyzinski2014seeded,kazemi2015can} minimize, but without proof 
of optimality. 
Corollary~\ref{singlecomm} establishes that 
when 
$g$ is an Erd\H{o}s-R\'enyi graph,
minimizing the edge 
mismatch cost function as defined by \cite{grossglauser11} 
indeed 
corresponds to the attacker's optimal strategy. This optimality result for the special case of Erd\H{o}s-R\'enyi graphs is also presented in a concurrent work \cite{daniel}. For the general case with two or 
more communities on the other hand, Theorem~\ref{mapthm} shows 
that 
inter-community and intra-community edge mismatches 
must be weighted differently.
\end{remark}

	\section{Probability of Error}
In the previous section, we 
characterized the attacker's optimal de-anonymization strategy for finite $n$. 	
	In this section, we derive sufficient conditions for which the attacker can find  \emph{asymptotically almost surely}  the
	 correct matching for the two community 
	 case. For notational simplicity let $p=p_{11}=p_{22} $ and $q=p_{12}$.  The corresponding result in the classical 
	 Erd\H{o}s-R\'enyi 
	 model follows as a special case.
	\begin{theorem} \label{Pethm}
	Given $g_1$, $g_2$ and $p,q \to 0$ where $q\leq p$, if  
	\begin{align*}
	\begin{split}
		s\left(1-\sqrt{1-s^2}\right)\left(p +2\cdot \frac{n_2}{n_1} q\right)=\frac{3\log n_1}{n_1}+\omega(n_1^{-1})\;,\\
		s\left(1-\sqrt{1-s^2}\right)\left(p +2\cdot \frac{n_1}{n_2} q\right)=\frac{3\log n_2}{n_2}+\omega(n_2^{-1})\;,	
	\end{split}	
	\end{align*}
		then the probability of error of the attacker can approach zero asymptotically, i.e., 
		$\argmin_{\pi} \Delta_{\pi} = \pi_0$ a.a.s. as $n_1,n_2 \to \infty$ where $ \pi_0 $ is the correct matching.  
	\end{theorem}
	\begin{proof} 	
		Note that in $\Delta_{\pi}$, 
		as $ p,q\to 0 $, $\frac{\omega_{11}}{\omega_{12}} \to 1$ if {$ p=\Theta(q) $}. 
		For tractability of our proof, we will assume
		that the attacker uses a sub-optimal cost function $ \Delta_\pi' $ where $ \frac{\omega_{11}}{\omega_{12}} = 1$.  
		Since the error probability with $\Delta_\pi' $
		upper bounds that with $\Delta_\pi$, we still obtain \emph{sufficient} conditions 
		for correct de-anonymization.	
		
		Let
		\begin{align*}
		S_{k_1,k_2}=\sum_{\pi\in\Pi_{k_1,k_2}}\ind\{\Delta'_\pi\leq\Delta'_0\}
		\end{align*}	
		where $\Delta'_0$ is the cost corresponding to $\pi_0$. Here $\Pi_{k_1,k_2}$ denotes the set of label assignments for nodes in $ g_2 $ 
		that is compatible with the given community assignment and where $k_1$ of the nodes in $ C_1 $ and $k_2$ of those in $ C_2 $ are mismatched. Note $ S_{k_1,k_2} $ denotes the number of node matchings with $ k_i $ mismatches in $ C_i $ that has mismatch cost not greater than that of true matching. Let us denote the number of labelings with at least one mismatch and edge cost not greater than $ \Delta'_0 $ with $ S $, that~is 
		\begin{align}
		S= \sum_{k_1=0}^{n_1}\;\sum_{k_2=0,\,k_1+k_2\neq 0}^{n_2} S_{k_1,k_2} \label{1942}
		\end{align}
		Our proof is based on showing the expected value of $ S $, a
		non-negative random variable, is asymptotically 0 if the conditions stated in the theorem are satisfied, therefore guaranteeing that an attacker using the $ \Delta'_\pi $ cost function would be able to recover the true matching. We have	
		\begin{align}
		\ex[S]&=\sum_{k_1}^{}\sum_{k_2}\sum_{\pi\in\Pi_{k_1,k_2}}\pr\{\Delta'_\pi-\Delta'_0\leq 0\}\nonumber\\
		&\leq\sum_{k_1}^{}\sum_{k_2}^{}|\Pi_{k_1,k_2}|\cdot \max_{\pi\in\Pi_{k_1,k_2}}\pr\{\Delta'_0-\Delta'_\pi\geq 0\}\label{0019}
		\end{align}	
		where the summations still have the same restrictions of \eqref{1942} but we omitted to keep the notation simple.	
		Note that based on the definitions of $ \Delta'_0 $ and $ \Delta'_\pi $, any node pair, $ e $, satisfying $ e=\pi(e) $ contributes equally to $ \Delta'_\pi $ and $\Delta'_0$,  whether they have an edge between them or not. Let us define the sets
		\begin{align*}
			\mc E_{intra}^\pi&=\{e\in (C_1\times C_1) \cup (C_2\times C_2)\;:\; e\neq\pi(e) \} \\
			\mc E_{inter}^\pi&=\{e\in C_1\times C_2 \;:\; e\neq\pi(e) \} 
		\end{align*}
		Note 
		\begin{align}
		|\mc E_{intra}^\pi|&=\sum_{i=1,2}\left[{k_i\choose 2} + k_i(n_i-k_i)\right]-|\mc E_{tr}^\pi| \label{0119} \\		
		|\mc E_{inter}^\pi|&=k_1 n_2 +k_2n_1-k_1k_2 \nonumber
		\end{align}    
		where the sum in \eqref{0119} is the number of intra-community node pairs having at least one mismatched node under $ \pi $ and $ \mc E_{tr}^\pi $ is the set of pairs that are transpositions of $ \pi $, that is pairs $ (a,b) $ with $ \pi(a)=b $ and $ \pi(b)=a $. We can write		
		\begin{equation*}
		\Delta'_0-\Delta'_\pi=Y_\pi-X_{\pi}
		\end{equation*}
		where 
		\begin{align*}
		\begin{split}
		Y_{\pi}=&\sum_{e\in \mc E_{intra}^\pi} |\ind\{e\in E_{g_1}\}-\ind\{e\in E_{g_2}\}|\\			
		+&\sum_{e\in \mc E_{inter}^\pi} |\ind\{e\in E_{g_1}\}-\ind\{e\in E_{g_2}\}| 	 
		\end{split}		
		\end{align*}
		and 
		\begin{align*}
		\begin{split}
		X_{\pi}=&\sum_{e\in\mc E_{intra}^\pi} |\ind\{e\in E_{g_1}\}-\ind\{\pi(e)\in E_{g_2}\}|\\
		+&\sum_{e\in \mc E_{inter}^\pi} |\ind\{e\in E_{g_1}\}-\ind\{\pi(e)\in E_{g_2}\}|\;.	 
		\end{split}
		\end{align*}	
		We can rewrite the difference between $ Y_\pi $ and $ X_\pi $ as
		\begin{align*}
		Y_\pi-X_\pi&=\sum_{e\in\mc E_{intra}^\pi}\hspace{-1em}A_e+ \sum_{e\in\mc E_{inter}^\pi}\hspace{-1em} B_e\quad\;\text{where}\\
		\begin{split}
			A_e=B_e&= |\ind\{e\in E_{g_1}\}-\ind\{e\in E_{g_2}\}|\\
			&-|\ind\{e\in E_{g_1}\}-\ind\{\pi(e)\in E_{g_2}\}|		
		\end{split}
		\end{align*}
		Note although the expressions for $ A_e $ and $ B_e $ are the same, their distributions are possibly different since they involve intra- and inter-community edges respectively.	Specifically,
		\begin{align*}
		\begin{split}
		u_1\triangleq\pr\{A_e=1\}=&\pr\left\{e\in E_{g_1},\pi(e)\in E_{g_2},e\notin E_{g_2} \right\}\\ 
		+ & \pr\left\{e\notin E_{g_1},\pi(e)\notin E_{g_2},e\in E_{g_2}\right\}
		\end{split}\\
		=&ps(1-s)
		\end{align*}
		By similar arguments
		\begin{align*}
		u_3&\triangleq \pr\{A_e=-1\}=ps(s+1-2ps)\\	
		v_1&\triangleq\pr\{B_e=1\}=qs(1-s)\\
		v_3&\triangleq\pr\{B_e=-1\}=qs(s+1-2qs) 
		\end{align*}
		and define $ u_2\triangleq 1-u_1-u_3 $, $ v_2\triangleq 1-v_1-v_3 $. 	If we denote $ Z=\sum_{e\in\mc E_{intra}^\pi}\hspace{-1em}A_e $ and $ T=\sum_{e\in\mc E_{inter}^\pi}\hspace{-1em}B_e $, we get	
		\begin{equation}
		\pr\{\Delta'_0-\Delta'_\pi\geq 0\}=\pr\{Y_\pi-X_{\pi}\geq 0\}=\pr\{Z+T\geq 0\}\label{0008}
		\end{equation}	
		Terms of the sum in $ Z $ can be dependent due to cycles in the mapping $ \pi $, and so are those of $ T $, but each term in $ Z $ is independent of all terms in $ T $. Next we decompose $ Z $ and $ T $ into three sums such that each sum consists of only independent terms. 
		\begin{prop}\label{partition}
			There exists a partition $ \mc E_{intra}^\pi=\cup_{i=1}^3 \mc E_{intra,i}^\pi$ such that 
		\begin{align}
			&\Big(\bigcup_{e\in \mc E_{intra,i}^\pi}\hspace{-1em} \{\pi(e)\}\Big) \cap \mc E_{intra,i}^\pi =\O \quad \text{and}\label{1859}\\
			&\left\vert\mc E_{intra,i}^\pi\right\vert \geq \left\lfloor \frac{\left\vert\mc E_{intra}^\pi\right\vert}{3}\right\rfloor \quad\text{for all $ i=1,2,3. $}\label{1901}
		\end{align}
		Similar result holds for $ \mc E_{inter}^\pi=\cup_{i=1}^3 \mc E_{inter,i}^\pi $.
		\end{prop}
\begin{proof}
	For any mapping $\pi$, we define a dependency graph $D^{\pi}$ such that every node pair in 
	$\mc E_{intra}^{\pi}$ corresponds to a vertex in $D^{\pi}$. An edge exists between $e$ and $e'$ 
	of $D^{\pi}$ if and only if $\pi(e) = e'$ or $\pi(e') = e$. Any partitioning of  
	$\mc E_{intra}^\pi$	that meets \eqref{1859}
	corresponds to a \emph{vertex coloring} in $D^{\pi}$. 
	Note that $ D^{\pi} $ consists of a finite number of disjoint cycles. Thus,  we can
	 color $ D^{\pi} $ using three colors (since an odd cycle would require three colors \cite{brooks}). 
	 The second condition of the Proposition \eqref{1901} follows from induction on the number
	  of cycles in $ D^{\pi} $.
\end{proof}
	Accordingly we let
		\[Z_i=\sum_{e\in \mc E_{intra,i}^\pi}\hspace{-1em} A_e\quad\text{and}\quad T_i=\sum_{e\in \mc E_{inter,i}^\pi}\hspace{-1em} B_e\] 	
		Continuing from \eqref{0008},
		\begin{align}
			&\pr\{Z+T\geq 0\}=\pr\{\textstyle\sum_{i=1}^{3} (Z_i + T_i)\geq 0 \}\nonumber\\
			&\leq \textstyle\sum_{i=1}^{3} \pr\{Z_i + T_i\geq 0 \}\leq 3\max_i\pr\{Z_i + T_i\geq 0 \} \label{0234}
		\end{align}
		For $ i=1,2,3 $
		\begin{align}
		&	\pr\{Z_i + T_i\geq 0 \}
		=\pr\{e^{\varphi (Z_i+T_i)}\geq 1\} \quad \varphi>0\nonumber\\
		&\phantom{\pr\{Z_i + T_i\geq 0 \}} \leq\ex[e^ {\varphi(Z_i+T_i)} ] \label{1823}\\
		&=\left(u_1e^\varphi+u_2+u_3e^{-\varphi}\right)^{n_{Z_i}}\left(v_1e^{\varphi}+v_2+v_3e^{-\varphi}\right)^{n_{T_i}}\nonumber\\		
		\begin{split}
		&\leq\exp\left[n_{Z_i}(u_3-e^\varphi u_1)(e^{-\varphi}-1)\right.
		\\&\hspace{3em}
		+\left.n_{T_i}(v_3-e^{\varphi} v_1)(e^{-\varphi}-1)\right] \label{155}
		\end{split}	
		\end{align}
	where $ n_{Z_i}=|\mc E_{intra,i}^\pi| $ and $ n_{T_i}=|\mc E_{inter,i}^\pi| $, \eqref{1823} is due to Markov's inequality and we use the inequality
	$x\leq e^{x-1}$
	in \eqref{155}. To find the smallest upper bound, we find the $ \varphi^* $ that sets the derivative of the exponent in \eqref{155}, which is a convex function of $ \varphi $, to 0. Inserting $ \varphi^* $ in the expression we get
	\begin{equation}
		\pr\{Z_i + T_i\geq 0 \}\leq e^{-\left(\sqrt{n_{Z_i}u_3+n_{T_i}v_3}-\sqrt{n_{Z_i}u_1+n_{T_i}v_1}\right)^2} \label{0219}
	\end{equation}	
	Now we find lower bounds on $ n_{Z_i} $ and $ n_{T_i} $
	\begin{align}
	n_{Z_i}&\geq \left\lfloor \frac{\left\vert\mc E_{intra}^\pi\right\vert}{3}\right\rfloor\hspace{-0.2em} \geq\hspace{-0.2em} \frac{1}{3}\sum_{i=1,2}\left[{k_i\choose 2} + k_i(n_i-k_i)\right]-\frac{|\mc E_{tr}^\pi|}{3}-1\nonumber\\
	&\geq  \frac{1}{3}\textstyle\sum_{i=1,2}k_i\left(n_i-\frac{k_i}{2}-1\right)-1 \quad\;\text{and}\label{2327}\\
		n_{T_i}&\geq \left\lfloor \frac{\left\vert\mc E_{inter}^\pi\right\vert}{3}\right\rfloor \geq \frac{k_1n_2+k_2n_1-k_1k_2}{3}-1\nonumber\\
		&\hspace{5.7em}\geq \frac{k_1n_2+k_2n_1}{6}-1\nonumber
	\end{align}
	\eqref{2327} is due to the bound $ |\mc E_{tr}^\pi|\leq (k_1+k_2)/{2} $. It can be checked that the derivative of the exponent in \eqref{0219} with respect to both $ n_{Z_i} $ and $ n_{T_i} $ is negative for sufficiently small  $ p $ and $ q $. Therefore from \eqref{0219} and the lower bounds over $ n_{Z_i} $ and $ n_{T_i} $ found above,
	\begin{align}
		\pr\{Z_i + T_i\geq 0 \}\leq e^{ -\frac{s}{3}(1-\sqrt{1-s^2})\left(k_1 n_1 p+ k_1n_2 q+k_2n_2p+ k_2n_1 q\right)}\label{0007}
	\end{align}	
	From \eqref{0008}, \eqref{0234} and \eqref{0007},
	\begin{align*}
		\pr\{\Delta'_0-\Delta'_\pi\geq 0\}\leq 3 e^{\frac{s}{3}(\sqrt{1-s^2}-1)\left(k_1 n_1 p+  k_1n_2 q+k_2n_2p+k_2n_1 q\right)}
	\end{align*}
	Let us now upper bound the other term in the summand of \eqref{0019},
	\begin{align}
	|\Pi_{k_1,k_2}|&\leq n_1^{k_1}n_2^{k_2}=\exp\left[k_1\log n_1+k_2\log n_2\right]	\label{0140}
	\end{align}	
	Using (\ref{0007}) and (\ref{0140}) in (\ref{0019}) we get,
	\begin{align*}	
	\begin{split}
	&\ex[S]\leq  \\
	&3\sum_{k_1}^{}\sum_{k_2}^{}\exp\left[k_1\left[\log n_1-s(1-\sqrt{1-s^2})(n_1p+ n_2 q)/3\right]\right.\\
	&+\left.k_2\left[\log n_2-s(1-\sqrt{1-s^2})(n_2p+ n_1 q)/3\right]\right]
	\end{split}	
	\end{align*}	
	The sum goes to 0 and thus we guarantee successful MAP de-anonymization if
	\begin{align*}
	\begin{split}
	&s\left(1-\sqrt{1-s^2}\right)\left(p + \frac{n_2}{n_1} q\right)=\frac{3\log n_1}{n_1}+\omega(n_1^{-1})\quad\text{and}\\
	&s\left(1-\sqrt{1-s^2}\right)\left(p + \frac{n_1}{n_2} q\right)=\frac{3\log n_2}{n_2}+\omega(n_2^{-1})\;.  
	\end{split}	
	\end{align*}	%
	\end{proof}	
%
	%
	\begin{corollary}\label{prob}
		In the case of single community (regular Erd\H{o}s-R\'enyi graph with $ n_1 $ nodes), a sufficient condition for de-anonymization is	
	\begin{align*}
	ps\left(1-\sqrt{1-s^2}\right)=\frac{3\log n_1}{n_1}+\omega(n_1^{-1}). 
	\end{align*}
	\end{corollary}
	\begin{proof}
		The result can be obtained by setting $ q=0 $, $ n_1=n_2$ in Theorem~\ref{Pethm}. Note that in this case there are no inter-community edges and the setting is equivalent to de-anonymization in each community separately. 
	\end{proof}
	\begin{remark}
 For the case of symmetric communities, $ n_1=n_2 $ and non-zero inter-community edge probability $ q $, conditions given in Theorem~\ref{Pethm} are less strict than the result in Corollary~\ref{prob} suggesting that inter-community edges help in the de-anonymization of nodes within a community.     		
	\end{remark}	
	\begin{remark}
		Concurrent work on the matching of regular (single-community) correlated Erd\H{o}s-R\'enyi graphs \cite{daniel} 
		proposes a proof based on combinatorial methods and obtains a sufficient condition for successful de-anonymization
		that is stronger than Corollary~\ref{prob} for the single community case. However our result in Theorem~\ref{Pethm} is more general since it handles graphs with community structure.
	\end{remark}	
	\begin{remark}
		Setting $s=1$ in Corollary~\ref{prob} provides the 
		following sufficient condition for successful de-anonymization:
		\begin{align}
		p=\frac{3\log n_1}{n_1}+\omega(n_1^{-1}) \label{auto}
		\end{align}
	Note for this special case, a necessary and sufficient condition for 
	successful de-anonymization  
	is the graph $ g $ not having any automorphisms (other than itself). 
	The condition for this, found in \cite{wright1971graphs}, is tighter than  
	\eqref{auto} by a factor of 3, showing that there is room for improvement in the calculation of probability of error.
	\end{remark}
	\section{Conclusion}	
	In this work we have investigated the de-anonymization problem in social networks for 
	graphs with community structure. We have characterized the optimal attack strategy 
	for this setting by determining the MAP estimate of the matching $\pi$, and determined 
	sufficient conditions for successful de-anonymization, asymptotically,  for large graphs. For the special case of a single community, our results have proved the 
	optimality of the de-anonymization strategy adopted in prior work. 

	\bibliographystyle{IEEEtran}
	\bibliography{refsPaper}

\end{document}